\documentclass{llncs}

\usepackage{times}
\usepackage{helvet}
\usepackage{courier}

\usepackage{ifpdf}
\ifpdf
\usepackage{microtype}
\fi

\usepackage{proof}
\usepackage{paralist}
\usepackage{enumerate}

\usepackage{amssymb}
\usepackage{amsmath}
\usepackage{multirow}
\usepackage{graphicx}

\usepackage{tikz}
\usetikzlibrary{arrows,positioning,automata,snakes}

\newcommand{\lp}{\ensuremath{\Pi}}
\newcommand{\alphabet}{\ensuremath{\mathcal{P}}}
\newcommand{\dneg}{\ensuremath{not\ }}
\newcommand{\head}[1]{\ensuremath{head(#1)}}
\newcommand{\body}[1]{\ensuremath{body(#1)}}
\newcommand{\pbody}[1]{\ensuremath{body(#1)^{+}}}
\newcommand{\nbody}[1]{\ensuremath{body(#1)^{-}}}
\newcommand{\atoms}[1]{\ensuremath{atom(#1)}}
\newcommand{\bodies}[1]{\ensuremath{body(#1)}}

\newcommand{\ass}{\ensuremath{\mathbf{A}}}
\newcommand{\Tass}{\ensuremath{\ass^\mathbf{T}}}
\newcommand{\Fass}{\ensuremath{\ass^\mathbf{F}}}
\newcommand{\tass}[1]{\ensuremath{\mathbf{T}#1}}
\newcommand{\fass}[1]{\ensuremath{\mathbf{F}#1}}
\newcommand{\es}[2]{\ensuremath{{ES}_{#2}(#1)}}

\newcommand{\dg}[1]{\ensuremath{{DG}(#1)}}
\newcommand{\loops}[1]{\ensuremath{{loop}(#1)}}
\newcommand{\scc}[1]{\ensuremath{scc(#1)}}

\newcommand{\up}[2]{\ensuremath{{UP}(#1,#2)}}
\newcommand{\wfn}[3]{\ensuremath{{WFN}[#1](#2,#3)}}
\newcommand{\wfj}[3]{\ensuremath{{WFJ}[#1](#2,#3)}}
\newcommand{\wfd}[3]{\ensuremath{{WFD}[#1](#2,#3)}}

\newcommand{\front}[1]{\ensuremath{{front}(#1)}}
\newcommand{\back}[1]{\ensuremath{{back}(#1)}}
\newcommand{\flow}[2]{\ensuremath{{SFG}(#1,#2)}}

\newcommand{\cassignment}{\ensuremath{A}}
\newcommand{\reachable}[3]{\ensuremath{{reachable}(#1,#2,#3)}}
\newcommand{\reach}[2]{\ensuremath{\text{REACH}[#1,#2]}}
\newcommand{\reached}[1]{\ensuremath{reached(#1)}}
\newcommand{\start}[1]{\ensuremath{start(#1)}}
\newcommand{\edge}[1]{\ensuremath{edge(#1)}}

\begin{document}

\abovedisplayshortskip=2pt 
\belowdisplayshortskip=2pt
\abovedisplayskip=2pt
\belowdisplayskip=2pt

\title{Efficient Approximation of Well-Founded Justification \\ and Well-Founded Domination \\[2em] (Corrected and Extended Version)}

\author{Christian Drescher \and Toby Walsh}

\institute{NICTA and the University of New South Wales}

\maketitle

\begin{abstract}
Many native ASP solvers exploit unfounded sets to compute consequences of a logic program via some form of well-founded negation, but disregard its contrapositive, well-founded justification (WFJ), due to computational cost.
However, we demonstrate that this can hinder propagation of many relevant conditions such as reachability.
In order to perform WFJ with low computational cost, we devise a method that approximates its consequences by computing dominators in a flowgraph, a problem for which linear-time algorithms exist.
Furthermore, our method allows for additional unfounded set inference, called well-founded domination (WFD).
We show that the effect of WFJ and WFD can be simulated for a important classes of logic programs that include reachability. \\[1em]

This paper is a corrected version of~\cite{drwa13a}. It has been adapted to exclude Theorem~10 and its consequences, but provides all missing proofs.
\end{abstract}

\section{Introduction}
The task of ASP solving is naturally broken up into a combination of search and propagation. The latter can be viewed in terms of inference operations like unit propagation on the Clark's completion~\cite{clark78} (UP) and unfounded set~\cite{gerosc91} computation. Unfounded sets characterise atoms in a logic program that might circularly support themselves when they have no external support and are thus not included in any answer set.
While propagating consequences from the completion is well studied and implemented\cite{gekanesc07a,gilima06a,linzha02a}, the task of efficiently propagating all information provided by unfounded sets is not yet solved~\cite{gebsch06c}. Instead, native ASP solvers\cite{gekanesc07a,lepffaeigopesc06a,siniso02a} apply unfounded set propagation asymmetrically via some form of well-founded negation (WFN), e.g., forward loop (FL) inference\cite{gebsch06c}, to exclude atoms that have no external support.
However, without their contrapositives, well-founded justification (WFJ) or its restriction, backward loop (BL) inference, as we shall see, propagation of many important conditions may be hindered. An example is given through reachability, which is relevant to a range of real world applications, and for which very natural and efficient ASP encodings exist.

In this paper, we address this deficiency.
Our main contribution is a linear-time algorithm that approximates the consequences of WFJ. The approach is based on a novel graph-representation of logic programs, termed the support flowgraph. We show that the problem of finding all dominators in such graph, for which efficient algorithms exist, can be used to approximate WFJ and even simulate BL and WFJ for important classes of logic programs.
Our techniques give rise to new forms of ASP inference, well-founded domination (WFD) and loop domination (LD). WFD and LD are the atom counterparts of WFJ and BL, respectively, i.e., they include atoms into an answer set in order to guarantee external support to already included atoms.
Then, we analyse the ASP inference on reachability. Contrary to the intuition that ASP systems naturally and efficiently handle reachability, we demonstrate that restricting inference to the combination of UP and WFN can hinder its propagation. Additional information, however, can be drawn from unfounded sets. We show that WFD and LD can lead to additional pruning, and that applying UP, FL, BL, and LD on reachability prunes all possible values.
%

\section{Preliminaries}
In this paper, we consider normal logic programs.
Given a set of atomic propositions~$\alphabet$, a \emph{(normal) logic program}~$\lp$ is a finite set of \emph{rules}~$r$ of the form $p_0 \leftarrow p_1 , \dots , p_m, \dneg p_{m+1} , \dots , \dneg p_n$ where $p_i\!\in\!\alphabet$ are \emph{atoms} ($0 \leq i \leq n$) and $\dneg{} p_j$ is the \emph{default negation} of~$p_j$ ($m < j \leq n$).
The atom~$\head{r} = p_0$ is referred to as the \emph{head} of~$r$ and the set $\body{r} = \{p_1 , \dots , p_m, \dneg p_{m+1} , \dots , \dneg p_n\}$ as the \emph{body} of~$r$. Let $\pbody{r} = \{p_1 , \dots , p_m\}$ and $\nbody{r} = \{p_{m+1} , \dots , p_n\}$. We denote by $\atoms{\lp}$ the set of atoms occurring in~$\lp$, and by $\bodies{\lp}$ the set of bodies in~$\lp$.
To access bodies sharing the same head~$p$, define $\body{p} = \{ \body{r}\!\mid\!r\!\in\!\lp,\ \head{r} = p \}$.
A set $X\!\subseteq\!\alphabet$ is an \emph{answer set} of a logic program~$\lp$ if $X$~is the least model of the \emph{reduct}~$\{ \head{r} \leftarrow \pbody{r}\!\mid\!r\!\in\!\lp,\ \nbody{r} \cap X = \emptyset\}$.

Answer sets can be characterised as assignments that assign true to an atom if and only if it is included in the answer set.
Extending assignments to bodies in a logic program can greatly reduce proof complexity~\cite{gebsch06c}.
Hence, for a logic program~$\lp$, we define an \emph{assignment}~$\ass$ as a set of \emph{literals} of the form~$\tass{x}$ or $\fass{x}$ where $x\!\in\!\atoms{\lp} \cup \bodies{\lp}$. Intuitively, $\tass{x}$~expresses that $x$~is assigned \emph{true} and $\fass{x}$~that it is~\emph{false} in~$\ass$.
The complement of a literal~$\sigma$ is denoted $\overline{\sigma}$. Let $\Tass = \{ x\!\mid\!\tass{x}\!\in\!\ass \}$ and $\Fass = \{ x\!\mid\!\fass{x}\!\in\!\ass \}$. $\ass$~is \emph{conflict-free} if~$\Tass \cap \Fass = \emptyset$, and it is \emph{body-saturated} if~$\{ \beta\!\in\!\bodies{\lp}\!\mid\!(\beta^+ \cap \Fass) \cup (\beta^- \cap \Tass) \neq \emptyset  \}\!\subseteq\!\Fass$, i.e., all bodies containing an atom that is assigned false must be false.
Finally, $\ass$~is \emph{total} if $\atoms{\lp} \cup \bodies{\lp} = \Tass \cup \Fass$. 

Following~Lee~\cite{lee05a}, answer sets of a logic program are given through total, conflict-free assignments that do not violate the conditions induced by the programs \emph{completion}~\cite{clark78}, and contain no non-empty unfounded set.
We use the concept of nogoods for representing the conditions from a program's completion. Following~\cite{gekanesc07a}, a \emph{nogood} is a set of literals $\delta = \{ \sigma_1, \dots, \sigma_n \}$, and given a set of nogoods~$\Delta$, a total and conflict-free assignment~$\ass$ is a \emph{solution} if~$\delta \not\subseteq \ass$ for all~$\delta\!\in\!\Delta$.
In our setting, every nogood is equivalent to a clause in CNF-SAT, e.g., the nogood $\delta = \{ \sigma_1, \dots, \sigma_n \}$ represents the clause $\overline{\sigma_1} \lor \dots \lor \overline{\sigma_n}$, and vice versa, and a set of nogoods is equivalent to a CNF-SAT formula.
To reflect the conditions from a program's completion, for $\beta = \{ a_1 , \dots , a_m, \dneg a_{m+1} , \dots , \dneg a_n \}\!\in\!\bodies{\lp}$, define
$$
\Delta_\beta = \left\{ \begin{array}{l}
\{\tass{a_1}, \dots, \tass{a_m}, \fass{a_{m+1}}, \dots \fass{a_n}, \fass{\beta} \}, \\
\{\fass{a_1}, \tass{\beta}\}, \dots, \{\fass{a_m}, \tass{\beta}\}, \{\tass{a_{m+1}}, \tass{\beta}\}, \dots, \{\tass{a_n}, \tass{\beta}\}
\end{array} \right\}
$$
and for an atom $p\!\in\!\atoms{\lp}$ with $\bodies{p} = \{\beta_1, \dots, \beta_k\}$ define
$$
\Delta_p = \left\{ \begin{array}{l} \{\tass{\beta_1}, \fass{p}\}, \dots, \{\tass{\beta_k}, \fass{p}\},  \{\fass{\beta_1}, \dots, \fass{\beta_k}, \tass{p} \} \end{array} \right\}\text{.}
$$
Intuitively, the nogoods in~$\Delta_\beta$ enforce the truth of body~$\beta$ if and only if all its members are satisfied, and the nogoods in~$\Delta_p$ enforce the truth of an atom~$p$ if and only if all at least one of its bodies is satisfied.
Let~$\Delta_{\lp} = \bigcup_{\beta \in \bodies{\lp}} \Delta_\beta \cup \bigcup_{p \in \atoms{\lp}} \Delta_p$. The solutions for~$\Delta_{\lp}$ correspond to the models of the completion of~$\lp$~\cite{gekanesc07a}. 

We now turn to unfounded sets.
For a logic program~$\lp$ and a set $U\!\subseteq\!\atoms{\lp}$, the \emph{external support} of~$U$ is defined as $\es{U}{\lp} = \{ \body{r}\!\mid\!r\!\in\!\lp,\ \head{r}\!\in\!U,\ \pbody{r} \cap U = \emptyset \}$.
Given an assignment~$\ass$, $U$~is an \emph{unfounded set}~\cite{gilima06a} of~$\lp$ w.r.t.~$\ass$ if $\es{U}{\lp}\!\subseteq\!\Fass$.
We define~$\ass$ as \emph{unfounded-free} if $\{ p\!\mid\!U\!\subseteq\!\atoms{\lp},\ p\!\in\!U,\ \es{U}{\lp}\!\setminus\!\Fass = \emptyset\}\!\subseteq\!\Fass$, i.e., all atoms from unfounded sets are false.
Attention is often restricted to unfounded sets that are subsets of strongly connected components (i.e., loops) in the \emph{dependency graph} of~$\lp$ given through~$\dg{\lp} = (\atoms{\lp} \cup \bodies{\lp}, \{ (\body{r}, \head{r})\!\mid\!r\!\in\!\lp\} \cup \{ (p, \body{r})\!\mid\!r\!\in\!\lp,\ p\!\in\!\pbody{r} \})$.
A non-empty set of atoms~$U \subseteq \atoms{\lp}$ is a \emph{loop} of~$\lp$ if for any $p,q\!\in\!U$ there is a path from~$p$ to~$q$ in~$\dg{\lp}$ such that all atoms in the path belong to~$U$~\cite{lee05a}. We denote by~$\loops{\lp}$ the set of all loops in~$\lp$, and define for~$\beta\!\in\!\bodies{\lp}$ the set~$\scc{\beta}$ as being composed of all atoms that belong to the same strongly connected component as~$\beta$.

Next, we introduce to propagation in ASP, starting with unit propagation~(UP).
Given an assignment~$\ass$, for a nogood~$\delta$ and~$\sigma\!\in\!\delta$, if~$\delta \!\setminus\! \{\sigma\}\!\subseteq\!\ass$ and $\overline{\sigma} \not\in \ass$ then~$\delta$ is \emph{unit} w.r.t.~$\ass$ and~$\overline{\sigma}$ is \emph{unit-resulting}, i.e., only unit-resulting literals can avert~$\delta\!\subseteq\!\ass$.
UP is the process of extending an assignment with unit-resulting literals. Formally, we define
\[
\up{\lp}{\ass} = \begin{cases}
\ass \cup \{\sigma\} & \text{if } \sigma \text{ is unit-resulting w.r.t. } \ass \text{ for some } \delta\!\in\!\Delta_{\lp}\text{,} \\
\ass & \text{otherwise.}
\end{cases}
\]
There might be several choices for~$\sigma$ in general. Therefore, we often consider fixpoints.
The effects of UP are determined by the nogoods in~$\Delta_{\lp}$, whose intuition was given earlier in this section.
In particular, fixpoint operation of UP achieves a body-saturated assignment.
Note that the notion of unit nogoods is the nogood-equivalent of unit clauses in CNF-SAT~\cite{mitchell05a}.
Hence, application of the unit clause rule (equivalently termed unit propagation) on the CNF-SAT representation of~$\Delta_{\lp}$ simulates UP on~$\Delta_{\lp}$.

An inference operation that aims at unfounded sets is well-founded negation~(WFN). WFN is the process of extending an assignment by assigning false to all atoms that are included in an unfounded set. Formally, for sets of atoms~$\Omega\!\subseteq\!2^{\atoms{\lp}}$ we define
\[
\wfn{\Omega}{\lp}{\ass} = \begin{cases}
\ass \cup \{\fass{p}\} & \text{if } U\!\in\!\Omega,\ p\!\in\!U,\ \es{U}{\lp}\!\setminus\!\Fass = \emptyset\text{,} \\
\ass & \text{otherwise.}
\end{cases}
\]
By construction, if $\Omega  = 2^{\atoms{\lp}}$ then fixpoint operation of $\wfn{\Omega}{\lp}{\ass}$ achieves an unfounded-free assignment.
In practice, it is enough to consider only unfounded sets that are loops, i.e., $\Omega = \loops{\lp}$, resulting in a restricted from of WFN referred to as forward loop (FL). Fixpoint operation of FL and UP, however, simulates the effect of WFN and UP~\cite{gebsch06c}.
FL can be implemented such that it takes $\mathcal{O}(|\lp|)$ time~(cf. \cite{angesc06a}).
The contrapositive of WFN is well-founded justification~(WFJ). It assigns true to the only external support of a set of atoms that contains a true atom.
\[
\wfj{\Omega}{\lp}{\ass} = \begin{cases}
\ass \cup \{\tass{\beta}\} & \text{if } U\!\in\!\Omega,\ p\!\in\!U \cap \Tass,\ \es{U}{\lp}\!\setminus\!\Fass = \{\beta\}\text{,} \\
\ass & \text{otherwise.}
\end{cases}
\]
Again, we consider the alternatives~$\Omega = 2^{\atoms{\lp}}$ (WFJ) and~$\Omega = \loops{\lp}$ (backward loop; BL).
In general, WFJ propagates more consequences than BL~\cite{gebsch06c}.
The time complexity of WFJ is bounded by~$\mathcal{O}(|\lp|^2)$, relatively high computational cost, as it amounts to failed-literal-detection and WFN.
Example~\ref{ex:inference} demonstrates the effect of WFJ.
\begin{example} \label{ex:inference}
Consider the logic program~$\lp$ given through the following set of rules.
\begin{align*}
a &\leftarrow \dneg{}b & c &\leftarrow d & e &\leftarrow f & c &\leftarrow a         \\
b &\leftarrow \dneg{}a & d &\leftarrow c & f &\leftarrow e & e &\leftarrow \dneg{} a
\end{align*}
Given the assignment~$\ass = \{\tass{c}\}$, applying UP to a fixpoint results in the (extended) assignment~$\{\tass{c}, \tass{\{c\}}, \tass{d}, \tass{\{d\}}\}$. WFN cannot propagate any additional information. In particular, neither UP nor WFN infer $\tass{a}$ (which is in the only total assignment that contains $\ass$ and corresponds to an answer set of~$\lp$). However, WFJ yields $\tass{\{a\}}$ since $\es{\{c,d\}}{\lp}\setminus\Fass = \{\{a\}\}$. In turn, repeated application of UP adds $\tass{a}$, $\fass{\{\dneg{}a\}}$ and $\fass{b}$ to the assignment, and WFN yields~$\fass{e}$ and $\fass{f}$ since $\es{\{e,f\}}{\lp}\setminus\{\{\dneg{} a\}\}= \emptyset$.
\end{example}
Now that we have established relevance of WFJ, we turn our attention to propagating WFJ. Recall that propagating WFJ can have quadratic costs. In the next section, we will introduce a method that approximates WFJ with only linear costs.

\section{Dominators in the Support Flowgraph} \label{sec:wfj}
We take a look at how support \emph{flows} through a logic program, represented in a flowgraph. In our context, a \emph{flowgraph} is a directed graph with a specially designated \emph{source} node.
\begin{definition}
Given a logic program~$\lp$ and an assignment~$\ass$. The \emph{support flowgraph} of~$\lp$ w.r.t.~$\ass$, denoted~$\flow{\lp}{\ass}$, is a directed graph defined as follows:
\begin{compactenum}
\item Create a node for each atom in $\atoms{\lp}$ and for each body in $\bodies{\lp}$, labelled with that atom or body, respectively.
\item The predecessors of an atom node~$p$ are all bodies in~$\bodies{p}\!\setminus\!\Fass$. The predecessors of a body node~$\beta$ are the set of atoms~$\phi(\beta)\!\setminus\!\Fass$ where
\[
\phi(\beta) = \begin{cases}
\beta^+ & \text{if } \scc{\beta} \cap \beta^+ = \emptyset \\
\scc{\beta} \cap \beta^+ & \text{otherwise.}
\end{cases}
\]
Observe that~$\phi(\beta)\!\subseteq\!\beta^+$.
\item Add a special node~$\top$ as the predecessor for all body nodes that do not have a predecessor, i.e., bodies from rules~$r\in\lp$ such that $\phi(\body{r}) = \emptyset$.
\end{compactenum}
\end{definition}
Nodes corresponding to atoms are referred to as \emph{atom nodes}, and nodes corresponding to bodies are referred to as \emph{body nodes}. We will also identify nodes with the atoms and bodies labelling them.
By construction, any predecessor of an atom is always a body, and for every body either all predecessors are atoms it positively depends on or the special node~$\top$ is the only predecessor. 
Note that~$\flow{\lp}{\ass}$ is a flowgraph with source node~$\top$. Its size is linear in the size of~$\lp$, and its construction can be made incremental w.r.t. the assignment, i.e., edges are removed down any branch of the search tree and re-inserted upon backtracking.

The intuition behind~$\flow{\lp}{\ass}$ is that (1)~the node~$\top$, representing syntactic truth, provides support to every non-false body that has no positive dependency, (2)~every body $\beta$~potentially provides external support to all atoms that appear in the head of a rule with body~$\beta$, and in turn, (3)~every non-false atom~$p$ can provide support to the bodies that are positively dependent on~$p$. The latter is determined by~$\phi$, according which, bodies in a non-trivial strongly connected component can only receive support from atoms that are in the same component. This design choice is motivated by the desire to restrict the intake of support to atoms in strongly connected components of the logic program.

It is easy to verify that if~$\ass$ is body-saturated then every body in~$\bodies{\lp}\!\setminus\!\Fass$ has a predecessor, and that by design, if~$\ass$ is unfounded-free then for every atom $p\!\in\!\atoms{\lp}\!\setminus\!\Fass$ there is a path from~$\top$ to~$p$.

We use cuts of the support flowgraph to analyse the flow of support.
For a directed graph~$(V,E)$ a \emph{cut}~$c=(S,W)$ is a partition of~$V$ into two disjoint subsets~$S$ and~$W$. 
For accessing the nodes in~$S$ that have an edge into~$W$, define $\front{c} = \{ u\!\in\!S\!\mid\!(u,v)\!\in\!E,\ v\!\in\!W \}$.
Note that, in principle, edges from~$W$ to~$S$ are allowed. For nodes in~$W$ that have an edge into~$S$, define $\back{c} = \{ u\!\in\!W\!\mid\!(u,v)\!\in\!E,\ v\!\in\!S \}$.
\begin{definition}
Given a logic program~$\lp$ and an assignment~$\ass$. A cut~$c = (S,W)$ of~$\flow{\lp}{\ass}$ is a \emph{support cut} if~$\top\!\in\!S$, $\front{c}\!\subseteq\!\bodies{\lp}$, and $\back{c}\!\subseteq\!\bodies{\lp}$.
\end{definition}
In words, for any support cut~$c = (S,W)$, the condition~$\front{c}\!\subseteq\!\bodies{\lp}$ ensures that whenever a body is in~$W$ then all its predecessors are in~$W$, and~$\back{c}\!\subseteq\!\bodies{\lp}$ ensures that whenever an atom is in~$W$ then all its successors are in~$W$.
\begin{example} \label{example:wfj} Consider the logic program~$\lp$ given through the following set of rules:
\noindent\parbox[][][t]{.45\linewidth}{
\begin{tikzpicture}
	[inner sep=1pt, >=stealth]
	\node (t) at (0.0,.5) {$\top$};
	\node (1) at (1.5,1) {$\{\dneg{} c\}$};
	\node (2) at (1.5,0) {$\{\dneg{} b\}$};
	\node (b) at (3.0,1) {$b$};
	\node (c) at (3.0,0) {$c$};
	\node (3) at (4.0,2.0) {$\{b,c\}$};
	\node (4) at (4.0,0.0) {$\{a\}$};
	\node (a) at (5.0,1) {$a$};
	\draw [->] (t) -- (1);
	\draw [->] (t) -- (2);
	\draw [->] (1) -- (b);
	\draw [->] (2) -- (c);
	\draw [->] (b) -- (3);
	\draw [->] (3) -- (a);
	\draw [->] (a) -- (4);
	\draw [->] (4) -- (b);
	\draw [dashed] plot [smooth] coordinates { (4,-0.5) (2.5,1) (3.5,2.5) };
	\draw (2.75,2.35) node {cut $c$};
	\draw [dashed] plot [smooth] coordinates { (5.5,1.5) (4.5,1.25) (3.,.5) (2.5,0) (2.5,-0.5) };
	\draw (5.2,1.75) node {cut $c^\prime$};
\end{tikzpicture}
}
\hfill
\noindent\parbox[][][t]{.5\linewidth}{
\begin{align*}
a &\leftarrow b,\ c & b &\leftarrow a         & b &\leftarrow \dneg{} c & c &\leftarrow \dneg{} b 
\end{align*}
Note that~$c\!=\!(\{\top, c, \{\dneg{} b\},\{\dneg{} c\}\},$ $\{a,b,\{a\},\{b,c\}\})$ and $c^\prime\!=\!(\{\top, b, c, \{b, c\}$, $\{\dneg{} b\}, \{\dneg{} c\}\}, \{a,c,\{a\}\})$ both are support cuts of $\flow{\lp}{\emptyset}$. Verify that~$\es{\{a,b\}}{\lp}\!=\!\{\{\dneg{} c\}\}\!=\!\front{c}$ and that $\es{\{a,c\}}{\lp}\!=\!\{\{\dneg{} b\}\} \subseteq \front{c^\prime} = \{\{\dneg{} b\}, \{b,c\}\}$.
}
Observe that $\front{c}$ represents external support of~$\{a,b\}$, while $\front{c^\prime}$ approximates (i.e., provides an upper bound of) the external support of~$\{a,c\}$.
\end{example}
Observe that~$\front{c} \cap \Fass = \emptyset$, by definition of a flowgraph.
The following lemma guarantees that every support cut in~$\flow{\lp}{\ass}$ separates a set of atoms from its external support.
\begin{lemma} \label{lemma:correctness}
Given a logic program~$\lp$ and a body-saturated assignment~$\ass$.
If~$c=(S,W)$ is a support cut of~$\flow{\lp}{\ass}$ then $\es{W \cap \atoms{\lp}}{\lp}\!\setminus\!\Fass\!\subseteq\!\front{c}$.
\end{lemma}
\begin{proof}
Let~$c=(S,W)$ be a support cut of~$\flow{\lp}{\ass}$. Then, $\front{c}\subseteq\bodies{\lp}\!\setminus\!\Fass$.
By definition of a support cut, for all~$r\in\lp$ such that $\body{r}\not\in\Fass$, if $\head{r}\!\in\!W\!\setminus\!\Fass$ then either $\body{r}\!\in\!\front{c}$ or $\body{r}\!\in\!W$. Since~$\ass$ is body-saturated, if $\body{r}\!\in\!W$ then $\phi(\body{r})\!\subseteq\!W \cap \atoms{\lp}$, and by definition of~$\phi$, $\pbody{r} \cap W \cap \atoms{\lp} \neq \emptyset$. In conclusion, we get $\es{W \cap \atoms{\lp}}{\lp}\!\setminus\!\Fass\!\subseteq\!\front{c}$.
\end{proof}
Hence, the set of bodies in~$\front{c}$ provide an upper bound on the external support of the atoms in~$W$.
However, we are more interested in finding support cuts that separate a set of atoms from a single external support, i.e., $\es{W \cap \atoms{\lp}}{\lp}\!\setminus\!\Fass = \{\beta\}$ for some~$\beta\!\in\!\bodies{\lp}$. Following from the previous lemma, this single external support is in a domination relationship with the set of atoms it supports.
Formally, in a flowgraph~$(V,E)$, a node~$u\!\in\!V$ \emph{dominates}~$v$ if every path from the source node to~$v$ passes through~$u$. It is easy to verify that a node~$v\!\in\!S$ dominates all nodes in~$W$ if and only if there is a cut~$c=(S,W)$ such that~$s\!\in\!S$ and~$\front{c}\!\subseteq\!\{v\}$.
\begin{theorem} \label{thm:correctness}
Given a logic program~$\lp$ and a body-saturated, unfounded-free assignment~$\ass$. Let~$U\!\subseteq\!\atoms{\lp}$ such that $U \cap \Tass \neq \emptyset$, and~$\beta\!\in\!\bodies{\lp}$.
If $\beta$~dominates all atoms in~$U$ in~$\flow{\lp}{\ass}$ then $\es{U}{\lp}\!\setminus\!\Fass = \{\beta\}$.
\end{theorem}
\begin{proof}
Let $\beta$~dominate all atoms in~$U$ in~$\flow{\lp}{\ass}$. Then, there is a support cut~$c=(S,W)$ such that~$\front{c}\!\subseteq\!\{\beta\}$ and~$U = W \cap \atoms{\lp}$. By Lemma~\ref{lemma:correctness}, $\es{W \cap \atoms{\lp}}{\lp}\!\setminus\!\Fass\!\subseteq\!\{\beta\}$ and therefore, $\es{U}{\lp}\!\setminus\!\Fass\!\subseteq\!\{\beta\}$. Since~$\ass$ is unfounded-free, it holds that~$\es{U}{\lp}\!\setminus\!\Fass \neq \emptyset$. We conclude~$\es{U}{\lp}\!\setminus\!\Fass = \{\beta\}$.
\end{proof}
The previous theorem grants the use of the domination relationship between bodies and atoms to compute consequences from WFJ.
\begin{example} Reconsider the logic program from Example~\ref{example:wfj}. The body~$\{\dneg{} c\}$ dominates the atom~$a$. Hence, if $a$ is assigned true then WFJ will set $\{\dneg{} c\}$ to true.\end{example}
A linear-time algorithm for finding all dominators in a flowgraph is provided in~\cite{geta04a}. It can be made incremental, i.e., few dominators might be recomputed at any stage during search, subject to removal and re-insertion of edges~\cite{srgale97a}. This puts our method to approximate WFJ on the same level of computational cost as WFN, resulting in a combined runtime complexity for unfounded set inference of~$\mathcal{O}(|\lp|)$.

The converse of Theorem~\ref{thm:correctness} does not hold in general, but we can provide conditions on logic programs for which our method is guaranteed to compute all consequences from WFJ and BL, respectively.
\begin{definition}
A \emph{unary logic program} is a logic program~$\lp$ such that for every rule~$r\!\in\!\lp$ it holds that $|\pbody{r}| \leq 1$. A \emph{component-unary logic program} is a logic program~$\lp$ such that for every rule~$r~\in~\lp$ it holds that $|\pbody{r} \cap \scc{\body{r}}| \leq 1$.
\end{definition}
Observe that every unary logic program is a component-unary logic program, but that component-unary logic programs are much more general.
A relevant example from the class of component-unary logic program is discussed in Section~\ref{sec:reach}.
In general, any logic program can become (component-) unary as truth values are assigned during search.
It is also important to note that for logic programs that are not (component-) unary, our method still simulates WFJ (BL) on the maximal (component-) unary sub-program.

For component-unary logic programs, the domination relationship between body- and atom nodes in the support flowgraph captures BL.
\begin{theorem} \label{thm:completion}
Given a component-unary logic program~$\lp$, and a body-saturated and unfounded-free assignment~$\ass$. Let~$L\!\in\!\loops{\lp}$ such that $L \cap \Tass \neq \emptyset$, and~$\beta\!\in\!\bodies{\lp}$.
The body~$\beta$ dominates all atoms in~$L$ in~$\flow{\lp}{\ass}$ if and only if~$\es{L}{\lp}\!\setminus\!\Fass = \{\beta\}$.
\end{theorem}
\begin{proof}
The implication $(\Rightarrow)$ holds by Theorem~\ref{thm:correctness}.
We have $(\Rightarrow)$ by Theorem~\ref{thm:correctness}.
It remains to show $(\Leftarrow)$ that if $\es{L}{\lp}\!\setminus\!\Fass = \{\beta\}$ then~$\beta$ dominates all atoms in~$L$ in~$\flow{\lp}{\ass}$.
Let~$\es{L}{\lp}\!\setminus\!\Fass = \{\beta\}$.
Construct a cut~$c=(S,W)$ of~$\flow{\lp}{\ass}$ where~$W = L \cup \{ \body{r}\!\mid\!r\in\lp,\ L \cap \pbody{r} \neq \emptyset \} \cup (\bodies{\lp} \cap \Fass)$, i.e., all atoms in~$L$ and all bodies that cannot provide external support to~$L$ are in~$W$, and all other nodes of~$\flow{\lp}{\ass}$, including~$\beta$, are in~$S$.
Then, for every $r\in\lp$ such that~$\body{r}\not\in\Fass$, if~$\head{r}\!\in\!W$ then either~$W \cap \pbody{r} \neq \emptyset$ or $\body{r} = \beta$.
Since~$\lp$ is component-unary, the only predecessor of a body node in~$W\!\setminus\!\Fass$ is an atom in~$W$, and since~$\es{L}{\lp}\!\setminus\!\Fass = \{\beta\}$, every predecessor of an atom node in~$W$ is either in~$W$ or equals~$\beta$, which is in~$S$. Hence, the cut~$c$ is a support cut with~$\front{c}\!\subseteq\!\{\beta\}$. By the assumption that~$\ass$ is body-saturated and unfounded-free, no node is disconnected and that every path in~$\flow{\lp}{\ass}$ from~$\top$ to any atom in~$L$ passes through~$\beta$.
In conclusion, the body node~$\beta$ dominates all atoms in~$L$ in~$\flow{\lp}{\ass}$.
\end{proof}
We can guarantee that our method simulates WFJ for unary logic programs.
\begin{theorem}
Given a unary logic program~$\lp$ and a body-saturated, unfounded-free assignment~$\ass$. Let~$U\!\subseteq\!\atoms{\lp}$ such that $U \cap \Tass \neq \emptyset$, and~$\beta\!\in\!\bodies{\lp}$.
If $\beta$~dominates all atoms in~$U$ in~$\flow{\lp}{\ass}$ if and only if $\es{U}{\lp}\!\setminus\!\Fass = \{\beta\}$.
\end{theorem}
\begin{proof}[Sketch]
The proof follows the one from Theorem~\ref{thm:completion}. Using the notation from its proof, the effect of~$\lp$ being a unary logic program is that for every set of atoms~$L\subseteq\atoms{\lp}$ the only predecessor of a body node in~$W$ is also an atom in~$W$.
\end{proof}
So far, we have restricted our attention to body nodes that dominate a set of atom nodes. In principle, however, any type of node can be a (strict) dominator. We will address dominators that are atom nodes in the next section.

\section{Well-Founded Domination}
We define an atom-equivalent of WFJ, that is, if a set of atoms~$U$ containing at-least one true atom, then any atom that appears positively in all external support of~$U$ must likewise be true.
\[
\wfd{\Omega}{\lp}{\ass} = \begin{cases}
\ass \cup \{\tass{p}\} &\!\!\text{if } U\!\in\!\Omega,\ q\!\in\!U \cap \Tass,\ \text{and} \\
                       &\!\!\es{U}{\lp}\!\setminus\!\Fass\!\subseteq\!\{ \body{r}\!\mid\!r\!\in\!\lp,\ p\!\in\!\pbody{r}\}\text{,} \\
\ass &\!\!\text{otherwise.}
\end{cases}
\]
As before, we consider the two alternatives~$\Omega\!=\!2^{\atoms{\lp}}$ (well-founded domination, WFD) and~$\Omega\!=\!{\loops{\lp}}$ (loop domination, LD).
We reuse the support flowgraph of a logic program and define a new form of cut to approximate consequences of WFD, following the strategy for approximating WFJ from the previous section.
\begin{definition}
Given a logic program~$\lp$ and an assignment~$\ass$. A cut~$c = (S,W)$ of~$\flow{\lp}{\ass}$ is an \emph{atom cut} if $\top\!\in\!S$, $\front{c}\!\subseteq\!\atoms{\lp}$, and $\back{c}\!\subseteq\!\bodies{\lp}$.
\end{definition}
The conditions~$\front{c}\!\subseteq\!\atoms{\lp}$ and $\back{c}\!\subseteq\!\bodies{\lp}$ for an atom cut~$c = (S,W)$ ensure that every predecessor and successor of an atom in~$W$ is also in~$W$.
Observe that,~$\front{c} \cap \Fass = \emptyset$ holds by definition of a flowgraph.
\begin{example} \label{example:wfd} Consider the logic program~$\lp$ given through the following set of rules:
\noindent\parbox[][][t]{.45\linewidth}
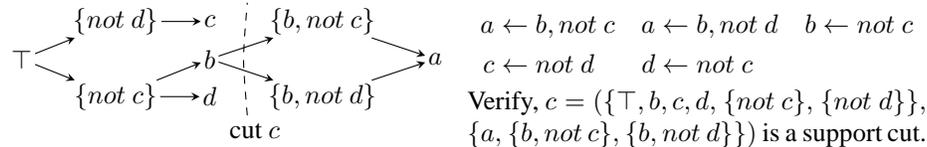
{
\begin{tikzpicture}
	[inner sep=1pt, >=stealth]
	\node (t) at (0.25,0.5) {$\top$};
	\node (1) at (1.5,1.0) {$\{\dneg{} d\}$};
	\node (2) at (1.5,0.0) {$\{\dneg{} c\}$};
	\node (c) at (2.75,1)   {$c$};
	\node (b) at (2.75,0.5) {$b$};
	\node (d) at (2.75,0.0) {$d$};
	\node (3) at (4.25,1.0) {$\{b,\dneg{} c\}$};
	\node (4) at (4.25,0.0) {$\{b,\dneg{} d\}$};
	\node (a) at (5.75,0.5) {$a$};
	\draw [->] (t) -- (1);
	\draw [->] (t) -- (2);
	\draw [->] (1) -- (c);
	\draw [->] (2) -- (b);
	\draw [->] (2) -- (d);
	\draw [->] (b) -- (3);
	\draw [->] (b) -- (4);
	\draw [->] (3) -- (a);
	\draw [->] (4) -- (a);
	\draw [dashed] plot [smooth] coordinates { (3.25,-0.25) (3.20,0.5) (3.25,1.25) };
	\draw (3.35,-0.45) node {cut $c$};
\end{tikzpicture}
}
\hfill
\noindent\parbox[][][t]{.5\linewidth}{
\begin{align*}
a &\leftarrow b, \dneg{}c & a &\leftarrow b, \dneg{} d & b &\leftarrow \dneg{} c \\ c &\leftarrow \dneg{} d & d &\leftarrow \dneg{} c 
\end{align*}
Verify,~$c = (\{\top, b, c, d,$ $\{\dneg{} c\},$ $\{\dneg{} d\}\},$ $\{a,$ $\{b,\dneg{}c\},$ $\{b,$ $\dneg{}d\}\})$ is a support cut.
}
Observe that $b$ appears positively in all external support of~$\{a\}$, i.e., $\es{\{a\}}{\lp}\!=\!\{\{b,\dneg{}c\},\{b,\dneg{}d\}\}\!\subseteq\!\{\body{r}\!\mid\!r\!\in\!\lp,\ \front{c} \in \pbody{r}\}$.
\end{example}
The following lemma guarantees that every atom cut in~$\flow{\lp}{\ass}$ separates a set of atoms~$U$ from the set of atoms that appear positively in the external support of~$U$.
\begin{lemma} \label{lemma:wfd:correctness}
Given a logic program~$\lp$ and a body-saturated assignment~$\ass$.
If~$c=(S,W)$ is an atom cut of~$\flow{\lp}{\ass}$ then $\es{W\!\cap\!\atoms{\lp}}{\lp}\!\setminus\!\Fass\!\subseteq\!\{ \body{r}\!\mid\!r\!\in\!\lp,\ \front{c} \cap \pbody{r} \neq \emptyset \}$.
\end{lemma}
\begin{proof}
Let~$c=(S,W)$ be an atom cut of~$\flow{\lp}{\ass}$. Let~$F = W \cap \{\body{r}\!\mid\!r\in\lp,\ \front{c} \cap \phi(\body{r}) \neq \emptyset\}$, the set of bodies in~$W$ that have a predecessor in~$\front{c}$. Construct a cut~$c^\prime=(S^\prime, W^\prime)$ where $S^\prime = S \cup F$ and $W^\prime = W \!\setminus\! F$, i.e., all bodies in~$F$ are shifted to $S$. Thus, for all~$\beta\!\in\!\front{c^\prime}$ it holds that~$\front{c} \cap \beta^+ \neq \emptyset$. Next, recall that in a support flowgraph, any predecessor of an atom node is always a body, i.e., no other node has a predecessor in~$\front{c}$. Hence, we get~$\front{c^\prime}\subseteq\bodies{\lp}$ and therefore, $c^\prime$~is a support cut of~$\flow{\lp}{\ass}$. By Lemma~\ref{lemma:correctness}, $\es{W^\prime \cap \atoms{\lp}}{\lp}\!\setminus\!\Fass\!\subseteq\!\front{c^\prime}$. By construction of~$c^\prime$ we have $W \cap \atoms{\lp} = W^\prime \cap \atoms{\lp}$, and conclude that~$\front{c} \cap \beta^+ \neq \emptyset$ for every~$\beta\!\in\!\es{W \cap \atoms{\lp}}{\lp}\!\setminus\!\Fass$.
\end{proof}
Hence, the atoms in~$\front{c}$ provide an upper bound on the atoms that appear positively in all external support of atoms in~$W$.
In order to guarantee that~$\front{c}$ represents the intersection of all external support, we restrict our attention to atom cuts with a single member in~$\front{c}$, i.e., dominators. Then, we can approximate WFD.
\begin{theorem} \label{thm:wfd:correctness}
Given a logic program~$\lp$ and a body-saturated assignment~$\ass$. Let~$U\!\subseteq\!\atoms{\lp}$ such that $U \cap \Tass \neq \emptyset$, and~$p\!\in\!\atoms{\lp}\!\setminus\!U$.
If~$p$ dominates all atoms in~$U$ in~$\flow{\lp}{\ass}$ then $\es{U}{\lp}\!\setminus\!\Fass\!\subseteq\!\{ \body{r}\!\mid\!r\in\lp,\ p\!\in\!\pbody{r} \}$.
\end{theorem}
\begin{proof}
Let~$p$ dominate all atoms in~$U$ in~$\flow{\lp}{\ass}$. Then, there is an atom cut~$c=(S,W)$ such that~$\front{c}\!\subseteq\!\{p\}$ and~$U = W \cap \atoms{\lp}$. By Lemma~\ref{lemma:wfd:correctness}, $\es{W \cap \atoms{\lp}}{\lp}\!\setminus\!\Fass\!\subseteq\!\{ \body{r}\!\mid\!r\in\lp,\ p\!\in\!\pbody{r} \}$ and therefore, $\es{U}{\lp}\!\setminus\!\Fass\!\subseteq\!\{ \body{r}\!\mid\!r\in\lp,\ p\!\in\!\pbody{r} \}$.
\end{proof}
\begin{example} Reconsider the logic program from Example~\ref{example:wfd}, where all external support of~$\{a\}$ contains~$b$, and~$b$ dominates~$a$. If $a$ is assigned true then WFD will set $b$ to true.\end{example}
Given a component-unary logic program, the following theorem guarantees that our technique can be used to simulate LD.
\begin{theorem} \label{thm:wfd:completion}
Given a component-unary logic program~$\lp$ and a body-saturated assignment~$\ass$. Let~$L\!\in\!\loops{\lp}$ such that~$L \cap \Tass \neq \emptyset$, and~$p\!\in\!\atoms{\lp}\!\setminus\!L$.
The atom node~$p$ dominates all atoms in~$L$ in~$\flow{\lp}{\ass}$ if and only if~$\es{L}{\lp}\!\setminus\!\Fass\!\subseteq\!\{ \body{r}\!\mid\!r\!\in\!\lp,\ p\!\in\!\pbody{r} \}$.
\end{theorem}
\begin{proof}
The implication $(\Rightarrow)$ holds by Theorem~\ref{thm:wfd:correctness}.
We have $(\Rightarrow)$ by Theorem~\ref{thm:wfd:correctness}.
Is remains to show $(\Leftarrow)$ that if $\es{L}{\lp}\!\setminus\!\Fass = \{ \body{r}\!\mid\!r\!\in\!\lp,\ p\!\in\!\pbody{r} \}$ then~$p$ dominates all atoms in~$L$ in~$\flow{\lp}{\ass}$.
Let~$\es{L}{\lp}\!\setminus\!\Fass\!\subseteq\!\{ \body{r}\!\mid\!r\!\in\!\lp,\ p\!\in\!\pbody{r} \}$.
Then, for every~$r\!\in\!\lp$ such that~$\body{r}\not\in\Fass$, if~$\head{r}\!\in\!L$ then either~$\pbody{r} \cap L \neq \emptyset$ or~$\pbody{r} = \{p\}$.
Construct a cut~$c = (S,W)$ of~$\flow{\lp}{\ass}$ where $W = L \cup \{\body{r}\!\mid\!r\!\in\!\lp,\ \pbody{r} \cap L \neq \emptyset \} \cup \{\body{r}\!\mid\!r\!\in\!\lp,\ p\!\in\!\pbody{r} \}$ and all other nodes of~$\flow{\lp}{\ass}$, including~$p$, are in~$S$. In particular, every predecessor of an atom node in~$W$ is in~$W$. Since~$\lp$ is component-unary, a predecessor of a body node in~$W$ is either an atom node in~$W$ or equals~$p$, which is in~$S$. Hence, the cut~$c$ is an atom cut with~$\front{c}\!\subseteq\!\{p\}$. In conclusion, $p$ dominates all atoms in~$L$ in~$\flow{\lp}{\ass}$.
\end{proof}
We can even simulate WFD if a unary logic program is given.
\begin{theorem}
Given a unary logic program~$\lp$ and a body-saturated assignment~$\ass$. Let~$U\!\subseteq\!\atoms{\lp}$ such that~$U \cap \Tass \neq \emptyset$, and~$p\!\in\!\atoms{\lp} \setminus U$.
The atom node~$p$ dominates all atoms in~$U$ in~$\flow{\lp}{\ass}$ if and only if~$\es{U}{\lp}\!\setminus\!\Fass\!\subseteq\!\{ \body{r}\!\mid\!r\!\in\!\lp,\ p\!\in\!\pbody{r} \}$.
\end{theorem}
\begin{proof}[Sketch]
The proof follows the one from Theorem~\ref{thm:wfd:completion}. Using the notation from its proof, the effect of~$\lp$ being a unary logic program is that for every set of atoms~$L\subseteq\atoms{\lp}$ the predecessor of a body node in~$W$ is an atom in~$W$~or equals~$p$.
\end{proof}

\section{Propagating Reachability in ASP} \label{sec:reach}
We want to analyse the impact of propagating ASP inference on the conditions represented by a logic program. These conditions are best studied in terms of constraints over finite domain variables (cf. CSP;\cite{robewa06a}).
Let~$V$ be a finite set of \emph{(domain) variables} where each variable~$v \in V$ has an associated finite domain~$dom(v)$. A \emph{constraint}~$c$ is a $k$-ary relation on the domains of $k$~variables given by~$scope(c)\!\in\!V^k$.
A \emph{(domain variable) assignment} is a function~$\cassignment$ that assigns to each variable a value from its domain. For an assignment~$\cassignment$, a constraint~$c$ is called \emph{domain consistent} if when any~$v\!\in\!scope(c)$ is assigned any value, there exist values in the domains of the variables in~$scope(c)\!\setminus\!\{ v\}$ such that $\cassignment(scope(c))\!\in\!c$, i.e., $c$~is \emph{satisfied}.
We will consider variables that represent a directed graph, called \emph{graph variables}, and sets of nodes, called \emph{node set variables}. Following~\cite{dodedu05a}, the domain of a graph variable is given via graph inclusion.
Graph inclusion defines a partial ordering among graphs, e.g., given two graphs $G = (V,E)$ and $G^\prime = (V^\prime,E^\prime)$, $G\!\subseteq\!G^\prime$ if $V\!\subseteq\!V^\prime$ and $E\!\subseteq\!E^\prime$.
Then, the domain of a graph variable~$v$ is defined as the lattice of graphs included between the greatest lower bound~$lb(v)$ and the least upper bound~$ub(v)$ of the lattice.
The domain of a node set variable is bounded by the subsets of nodes in the graph, and we denote the greatest lower bound by~$lb(v)$ and the least upper bound by~$ub(v)$.
If for a domain variable~$v$ the associated domain is a singleton, we say that $v$~is \emph{fixed} and simply write~$v$ instead of~$lb(v)$ or~$ub(v)$.

Reachability is a relevant condition in many ASP applications.
Given a graph variable~$G$, and node set variables~$S$ and~$N$, the constraint~$\reachable{G}{S}{N}$ states that $N$~is the set of nodes reachable from some node in~$S$, i.e., the subgraph induced by~$N$ is connected.
For encoding reachability into ASP we use atoms of the form~$\edge{Y,X}$, $\start{X}$, and $\reached{X}$ to capture the membership of edges in~$G$, and nodes in~$S$ and~$N$, respectively. Nodes in~$G$ are given implicitly through the edges in~$G$. We denote by~$\reach{G}{S}$ the following rules:
\begin{align*}
\forall\ X\!\in\!ub(S): & \hspace{1cm} \reached{X} \leftarrow \start{X} \\
\forall\ (Y,X)\!\in\!ub(G): & \hspace{1cm} \reached{X} \leftarrow \reached{Y},\ \edge{Y,X}
\end{align*}
We assume that rules for~$\edge{Y,X}$ and $\start{X}$ are provided elsewhere, as we restrict our attention to reachability.
It is easy to verify that a node~$t\!\in\!N$ if and only if $\tass{\reached{t}}$~is in an assignment representing an answer set of the resulting program.
In the following, we study the impact of propagation on~$\reach{G}{S}$ in terms of consistency on reachability.
We start with the special case where~$G$ and~$S$ are fixed.
UP and FL on~$\reach{G}{S}$ prunes all values of~$N$.
\begin{theorem} \label{wfn:dc}
If~$G$ and~$S$ are fixed, then UP and FL on~$\reach{G}{S}$ achieve domain consistency on~$\reachable{G}{S}{N}$.
\end{theorem}
\begin{proof}
Assume UP and FL reached the fixpoint~$\ass$, and~$\ass$ is conflict-free.
Let~$v \in G$.
If~$\fass{\reached{v}}\not\in\ass$ then $\es{\{v\}}{\reach{G}{S}}\setminus\Fass \neq \emptyset$. Hence, either~$\fass{\{\start{v}\}} \not\in\ass$ or $\fass{\{\reached{u}, \edge{u,v}\}} \not\in\ass$ for some $(u,v) \in G$, i.e., either $v \in S$ or $v$~has a predecessor~$u$ that is reached. By successively applying the same argument, we obtain loops, each of which concludes in a start node. Hence, there is an assignment with~$v \in N$ satisfying the constraint.
On the other hand, if~$\tass{\reached{v}}\not\in\ass$ then the nogood $\{\fass{\reached{v}}, \tass{\{\start{v}\}}\} \in \Delta_{\reached{v}}$ guarantees that $v \not\in S$. Similarly, for every~$(u,v) \in G$, the nogood $\{\fass{\reached{v}}, \tass{\{\reached{u}, \edge{u,v}\}}\} \in \Delta_{\reach{G}{S}}$ guarantees that every predecessor~$u$ is disconnected. Moreover, the atoms in each loop~$L$ starting from~$v$ are either disconnected, i.e., we have $\es{L}{\reach{G}{S}}\setminus\Fass = \emptyset$, or (since the graph is fixed) their subsets are guaranteed external support via a path that does not go through~$v$.
Hence, there is an assignment with~$v \not\in N$ satisfying the constraint.
In conclusion, $\reachable{G}{S}{N}$ is domain consistent.
\end{proof}
We now turn our attention to another special case of~$\reachable{G}{S}{N}$, that is, the value of~$N$ is fixed.
Then, UP and WFN on~$\reach{G}{S}$ can hinder propagation, in general, and the construction of a counter example is easy.
However, we can guarantee that the addition of WFJ inference prunes all values.
\begin{theorem} \label{wfj:dc}
If $N$ is fixed then UP and BL on~$\reach{G}{S}$ achieve domain consistency on~$\reachable{G}{S}{N}$.
\end{theorem}
\begin{proof}
Assume UP and BL result in the fixpoint~$\ass$, and~$\ass$ is conflict-free.
Let~$(u,v)\!\in\!ub(G)$.
If~$\tass{\edge{u,v}}\!\not\in\!\ass$ the nogood $\{ \tass{\{ \reached{u}, \edge{u,v}\}}, \fass{\reached{v}} \}\!\in\!\Delta_{\reach{G}{S}}$ guarantees that $(u,v)$~does not connect a node that is reached with a disconnected one. Hence, there is an assignment with~$(u,v) \in G$ satisfying the constraint.
On the other hand, if~$\fass{\edge{u,v}} \not\in \ass$ then $\es{\{v\}}{\reach{G}{S}}\setminus\Fass\!\neq\!\{ \{\reached{u}, \edge{u,v}\} \}$, i.e., if $v$~is reached then either~$v\!\in\!S$ or there is some other edge that can connect a reached node to~$v$. By successively applying the same argument, we obtain loops, each of which concludes in a node from~$S$. Hence, there is an assignment with~$(u,v)\!\not\in\!G$ satisfying the constraint.
The proof for any~$v\!\in\!ub(S)$ follows similar arguments. We conclude that $\reachable{G}{S}{N}$ is domain consistent.
\end{proof}
If the value of~$N$ is not fixed, however, domain consistency is not guaranteed. (Again, the construction of a counter example is easy.) Additional pruning is required. We can show that UP, FL, BL, LD, altogether propagate reachability efficiently.
\begin{theorem}
Propagating UP, FL, BL and LD on~$\reach{G}{S}$ achieves domain consistency on~$\reachable{G}{S}{N}$.
\end{theorem}
\begin{proof}[Sketch]
Assume UP, FL, BL and LD reached the fixpoint~$\ass$, and~$\ass$ is conflict-free.
For any edge $(u,v) \in ub(G)$, the proof follows the one for Theorem~\ref{wfj:dc}, i.e., UP ensures that when assigning~$(u,v) \in G$ the edge does not connect a node that is reached with a disconnected one, and BL guarantees that when assigning~$(u,v) \not\in G$ there is some other way to connect to~$v$ if $v \in N$.
Similarly, for any node $v \in ub(S)$, UP ensures that when assigning~$v \in S$ the node $v$ is reached, and BL guarantees that when assigning~$v \not\in S$ there is some path connecting a start node to~$v$ if $v \in N$.
Moreover, following the arguments in the proof of Theorem~\ref{wfn:dc}, FL removes nodes from~$ub(N)$ that cannot be reached in any satisfying domain variable assignment, and for every node~$v \in ub(N)$, if~$\tass{\reached{v}}\not\in\ass$ then $v \not\in lb(S)$ and every predecessor can be disconnected.
It remains to show that if~$\tass{\reached{v}}\not\in\ass$ then the atoms in each loop~$L$ starting from~$v$ can be disconnected or reached via some path that does not go through~$v$. This is guaranteed by LD, i.e., we have $\es{L}{\lp}\!\setminus\!\Fass\!\not\subseteq\!\{ \body{r}\!\mid\!r\!\in\!\lp,\ \reached{v}\!\in\!\pbody{r} \}$.
Hence, there is an assignment with~$v \not\in N$ satisfying the constraint.
In conclusion, $\reachable{G}{S}{N}$ is domain consistent.
\end{proof}
While previous theorems establish practical relevance of BL and LD, recall that, to our knowledge, existing ASP solvers to not implement BL and LD.
However, our efficient approximations of WFJ and WFD can be used to simulate BL and LD, respectively, since $\reach{G}{S}$ results in a component-unary logic program.
If, in addition, the value of~$G$ is fixed such that all atoms in $\reach{G}{S}$ of the form~$\edge{Y,X}$ can be dropped, resulting in a unary logic program, then our method even simulates WFJ and WFD, respectively.

\paragraph{Corrigendum}
An earlier version of this paper~\cite{drwa13a} includes a theorem stating that there is no polynomial size logic programming encoding of reachability such that UP achieves domain consistency. The proof combined (1) a famous result from Karchmer and Wigderson showing that the smallest fan-out-1 monotone Boolean circuit for reachability is super-polynomial in the number of vertices in a graph~\cite{kawi90}, and (2) a result from Bessiere et al.~\cite{bekanawa09a} showing that there is a polynomial size encoding of a constraint into CNF-SAT such that applying unit propagation achieves domain consistency on the constraint if and only if a domain consistency propagator for the constraint can be computed by a polynomial size monotone circuit. However, this does not prove the theorem since fan-out-2 (monotone) circuits can achieve exponential compression in terms of size complexity over fan-out-1 (monotone) circuits~(cf.~\cite{savage98}).
In fact, one can construct a polynomial size logic programming encoding of reachability based on the idea of Warshall's algorithm for transitive closure of directed graphs~\cite{warshall62}.
%

\section{Experiments}
Implementing Tarjan's linear-time algorithm for finding all dominators in a flowgraph~\cite{geta04a} is a challenging engineering exercise as it relies on sophisticated data structures.
Hence, for practical reasons, we have integrated BL into the ASP solver~\emph{clasp}~(2.1.1) via failed-literal-detection and FL. This has high computational costs.
To compare with the state-of-the-art, i.e., using only UP and FL, we include the default setting of \emph{clasp} in our analysis.
We conducted experiments on search problems that make use of reachability conditions. Our benchmarks stem from the Second ASP Competition~\cite{devebogetr09a}.
The following definitions apply to Table~\ref{table:results} of results.
\textbf{UP+FL} denotes \emph{clasp}'s default setting, and \textbf{UP+FL+BL} denotes the setting that integrates BL.
In each benchmark class, \textbf{\#N} denotes the total number of instances and \textbf{\#S} denotes the number of instances for which the program terminated within the allowed time. \textbf{Time} denotes the time taken to compute all instances in the class that were solved in both settings. Similarly, \textbf{\#B} denotes the total number of branches and \textbf{\#C} denotes the number of conflicts during search, aggregated over all instances in the class that were solved in both settings.
Experiments were run on a Linux~PC, where each run was limited to 1200s time on a single 2.00~GHz core.

From the results shown in Table~\ref{table:results}, it can be concluded that information from BL prunes search dramatically: The additional propagation in~\textbf{UP+FL+BL} decreases the number of branches and conflicts by roughly one order of magnitude in comparison to~\textbf{UP+FL}. 
On the other hand, high computational costs of propagating BL via failed-literal-detection are clearly reflected in the run times of~\textbf{UP+FL+BL}.
These costs, however, can be drastically reduced by using Tarjan's linear-time algorithm, and by making the computation of dominators incremental. In conclusion, our experiments encourage the implementation of our techniques.
\begin{table}[t]
\centering
\caption{Experimental Data. \label{table:results}}
\vspace{-1em}
\begin{tabular}{lc@{\hspace{.4cm}}crrr@{\hspace{.4cm}}crrr}
\hline
                         &              & \multicolumn{4}{c@{\hspace{.4cm}}}{\textbf{UP+FL}}                         & \multicolumn{4}{c}{\textbf{UP+FL+BL}} \\
\textbf{Benchmark Class} & \textbf{\#N} & \multicolumn{1}{c}{\textbf{\#S}} & \multicolumn{1}{c}{\textbf{Time}} & \multicolumn{1}{c}{\textbf{\#B}} & \multicolumn{1}{c@{\hspace{.4cm}}}{\textbf{\#C}} & \multicolumn{1}{c}{\textbf{\#S}} & \multicolumn{1}{c}{\textbf{Time}} & \multicolumn{1}{c}{\textbf{\#B}} & \multicolumn{1}{c}{\textbf{\#C}} \\
\hline
Connected Dominating Set          & 21 & 20 &  202 & 11320.9k &  6339.2k & 20 & 3342 & 6887.4k & 3655.0k \\
Generalised Slitherlink           & 29 & 29 &    3 &    22.3k &     4.7k & 29 &    4 &    1.3k &    0.4k \\
Graph Partitioning                & 13 & 13 &  147 &  3159.4k &  2344.5k & 13 &  785 & 1138.5k &  810.2k \\
Hamiltonian Path                  & 29 & 29 &    1 &    44.0k &    17.9k & 29 &    8 &    6.1k &    2.9k \\
Maze Generation                   & 29 & 26 &   53 &  3831.5k &  1906.4k & 20 & 1700 & 1425.8k &  880.8k \\
\hline
\end{tabular}
\vspace*{-2em}
\end{table}

\section{Related Work}
A straightforward way of computing answer sets of logic programs is a reduction to CNF-SAT.
This may require the introduction of additional atoms. As shown by Lifschitz and Razborov~\cite{lifraz04a}, it is unlikely that, in general, a polynomial-size translation from ASP to CNF-SAT would not require additional atoms.
Evidence is provided by the encoding of Lin and Zhao~\cite{linzha02a} that has exponential space complexity.
Another result, shown by Niemel{\"a}~\cite{niemela99a}, is that ASP cannot be translated into CNF-SAT in a \emph{faithful} and \emph{modular} way.
Reductions based on level-mappings devised in~\cite{jani11a} are non-modular but can be computed systematically, using only sub-quadratic space.
%
%
%
%
An advantage of native ASP solvers like \emph{clasp}~\cite{gekanesc07a}, \emph{dlv}~\cite{lepffaeigopesc06a}, and \emph{smodels}~\cite{siniso02a} over SAT-based systems~\cite{gilima06a,jani11a,linzha02a}, however, is that they can potentially propagate more consequences, e.g., using our techniques, faster.

Formal means for analysing ASP computations in terms of inference were introduced by~Gebser and Schaub~\cite{gebsch06c}.
According which, \emph{smodels' atmost} and \emph{dlv's greatest unfounded set detection}, both compute WFN and FL, respectively.
Similarly, \emph{clasp's unfounded set check} computes~FL~\cite{angesc06a}.
Gebser and Schaub also identified the backward propagation operations for unfounded sets, i.e., WFJ and BL.
A method that can be used to propagate BL has been proposed by Chen, Ji and Lin~\cite{chjili13a}, but it is inefficient due to high computational costs.
We have devised a linear-time approximation of WFJ and shown under which conditions our method simulates WFJ and BL, respectively.
Moreover, we have put forward WFD and LD as new forms of inference that can draw additional consequences from unfounded sets.
Our approach uses a reduction to the task of finding all dominators in the support flowgraph of a logic program, for which efficient algorithms exist.
For instance, Tarjan's algorithm~\cite{geta04a} runs in linear time, and computing all dominators can be made incremental~\cite{srgale97a}.

\section{Conclusions}
Our work is motivated by the desire to understand the effect of propagation in ASP and the diverse modelling choices that arise from logic programming on the process of solving a combinatorial problem.
In this paper, we have analysed the ASP inference on reachability conditions. In contrast to the intuition that the inference in extisting ASP systems naturally handles reachability, our results have shown that even some restricted variants of reachability cannot be efficiently propagated by a combination of UP and WFN. This gap can be closed with WFJ, but existing implementations are inefficient. Our main contribution was a linear-time approximation of WFJ based on a reduction to finding all dominators in a flowgraph representation of the logic program. This gave rise to novel forms of inference, WFD and DL, which can be approximated using the same techniques. We have outlined classes of logic programs for which our approximations simulate WFJ and BL, and WFD and LD, respectively. This includes reachability.
Our experimental data encourages the integration of an incremental linear-time algorithm for finding all dominators into an ASP system.
Despite our best efforts, efficient algorithms for fully propagating WFJ and WFD 
remain an open problem.

\paragraph{Acknowledgements}
NICTA is funded by the Australian Government as represented by the Department of Broadband, Communications and the Digital Economy and the Australian Research Council through the ICT Centre of Excellence program.
The authors would like to thank Valentin Mayer-Eichberger for useful discussions, and for pointing out the error in Theorem~10 of an earlier version of this paper.


\begin{thebibliography}{10}

\bibitem{angesc06a}
Anger, C., Gebser, M., Schaub, T.: Approaching the core of unfounded sets. In:
  Proceedings of NMR'06. pp. 58--66 (2006)

\bibitem{bekanawa09a}
Bessi{\`e}re, C., Katsirelos, G., Narodytska, N., Walsh, T.: Circuit complexity
  and decompositions of global constraints. In: Proceedings of IJCAI'09. pp.
  412--418 (2009)

\bibitem{chjili13a}
Chen, X., Ji, J., Lin, F.: Computing loops with at most one external support
  rule. ACM Trans. Comput. Logic  14(1),  3:1--3:34 (2013)

\bibitem{clark78}
Clark, K.: Negation as failure. In: Logic and Data Bases, pp. 293--322. Plenum
  Press (1978)

\bibitem{devebogetr09a}
Denecker, M., Vennekens, J., Bond, S., Gebser, M., Truszczy{\'n}ski, M.: The
  second answer set programming competition. In: Proceedings of LPNMR'09. pp.
  637--654 (2009)

\bibitem{dodedu05a}
Dooms, G., Deville, Y., Dupont, P.: {CP(Graph)}: Introducing a graph
  computation domain in constraint programming. In: Proceedings of CP'05. pp.
  211--225. Springer (2005)

\bibitem{drwa13a}
Drescher, C., Walsh, T.: Efficient approximation of well-founded justification
  and well-founded domination. In: Proceedings of LPNMR'13. Springer (2013), to
  appear

\bibitem{gekanesc07a}
Gebser, M., Kaufmann, B., Neumann, A., Schaub, T.: Conflict-driven answer set
  solving. In: Proceedings of IJCAI'07. pp. 386--392. AAAI Press/MIT Press
  (2007)

\bibitem{gebsch06c}
Gebser, M., Schaub, T.: Tableau calculi for answer set programming. In:
  Proceedings of ICLP'06. pp. 11--25. Springer (2006)

\bibitem{geta04a}
Georgiadis, L., Tarjan, R.E.: Finding dominators revisited. In: Proceedings of
  SODA'04. pp. 869--878. SIAM (2004)

\bibitem{gilima06a}
Giunchiglia, E., Lierler, Y., Maratea, M.: Answer set programming based on
  propositional satisfiability. Journal of Automated Reasoning  36(4),
  345--377 (2006)

\bibitem{jani11a}
Janhunen, T., Niemel{\"a}, I.: Compact translations of non-disjunctive answer
  set programs to propositional clauses. In: Proceedings of LPNMR'11. pp.
  111--130. Springer (2011)

\bibitem{kawi90}
Karchmer, M., Wigderson, A.: Monotone circuits for connectivity require
  super-logarithmic depth. {SIAM} Journal on Discrete Mathematics  3(2),
  255--265 (1990)

\bibitem{lee05a}
Lee, J.: A model-theoretic counterpart of loop formulas. In: Proceedings of
  IJCAI'05. pp. 503--508. Professional Book Center (2005)

\bibitem{lepffaeigopesc06a}
Leone, N., Pfeifer, G., Faber, W., Eiter, T., Gottlob, G., Perri, S.,
  Scarcello, F.: The dlv system for knowledge representation and reasoning. ACM
  Trans. Comput. Log.  7(3),  499--562 (2006)

\bibitem{lifraz04a}
Lifschitz, V., Razborov, A.: Why are there so many loop formulas? ACM
  Transactions on Computational Logic  7(2),  261--268 (2006)

\bibitem{linzha02a}
Lin, F., Zhao, Y.: {ASSAT}: Computing answer sets of a logic program by {SAT}
  solvers. In: Proceedings of AAAI'02. pp. 112--118. AAAI Press/MIT Press
  (2002)

\bibitem{mitchell05a}
Mitchell, D.: A {SAT} solver primer. Bulletin of the European Association for
  Theoretical Computer Science  85,  112--133 (2005)

\bibitem{niemela99a}
Niemel{\"a}, I.: Logic programs with stable model semantics as a constraint
  programming paradigm. Annals of Mathematics and Artificial Intelligence
  25(3-4),  241--273 (1999)

\bibitem{robewa06a}
Rossi, F., {van Beek}, P., Walsh, T. (eds.): Handbook of Constraint
  Programming. Elsevier (2006)

\bibitem{savage98}
Savage, J.E.: Models of Computation. Addison-Wesley (1998)

\bibitem{siniso02a}
Simons, P., Niemel{\"a}, I., Soininen, T.: Extending and implementing the
  stable model semantics. Artificial Intelligence  138(1-2),  181--234 (2002)

\bibitem{srgale97a}
Sreedhar, V.C., Gao, G.R., Lee, Y.F.: Incremental computation of dominator
  trees. ACM Trans. Program. Lang. Syst.  19(2),  239--252 (1997)

\bibitem{gerosc91}
{Van Gelder}, A., Ross, K.A., Schlipf, J.S.: The well-founded semantics for
  general logic programs. Journal of the ACM  38(3),  620--650 (1991)

\bibitem{warshall62}
Warshall, S.: A theorem on boolean matrices. Journal of the ACM  9(1),  11--12
  (1962)

\end{thebibliography}
\end{document}